\documentclass[a4paper,12pt]{article}

\usepackage{color, comment,url}

\usepackage{geometry}    
\usepackage{amsmath, amssymb, amsthm, amsfonts}  
\usepackage{graphicx, subcaption}  
\usepackage{natbib}  
\usepackage{enumerate, paralist}  
\usepackage{arydshln, umoline}  
\usepackage{appendix}  
\usepackage{newtxtext, newtxmath}  
\usepackage{setspace}  
\usepackage{mathtools}  
\usepackage{physics}  
\usepackage{color}  

\newtheorem{proposition}{Proposition}

\theoremstyle{definition}

\theoremstyle{remark}

\begin{document}

\title{Shinohara Rock-Paper-Scissors\thanks{I thank Shinpei Noguchi and Takahiro Watanabe for their valuable comments. I also thank the co-editor and an anonymous reviewer for their helpful suggestions regarding the clarity of presentation and the pedagogical aspect of Shinohara RPS.  I am grateful to the radio program After 6 Junction, which introduced Shinohara RPS and inspired part of my interest in the topic.}}

\author{Takashi Ui\thanks{Kanagawa University and Hitotsubashi University. E-mail: oui@econ.hit-u.ac.jp.}}

\date{April 2025}
\maketitle

\begin{abstract}
This paper analyzes Shinohara Rock-Paper-Scissors (RPS), a variant of the classic RPS game introduced by board game designer Yoshiteru Shinohara. Players compete against a host who always plays rock, so they choose either rock or paper. The twist is that if two or more players choose paper, they are eliminated, and the last remaining player is the winner, creating strategic tension among the players. There exists a unique symmetric subgame perfect equilibrium, in which the probability of choosing paper satisfies the equation $(1-p)^{n-1} + p^{n-1}/n = 1/n$ when $n$ players remain. The game also admits a continuum of asymmetric equilibria. 
\end{abstract}

\newpage
\section{Introduction}

The Rock-Paper-Scissors (RPS) game offers an interactive method for selecting a prize winner from an audience at an event.
All contestants are asked to stand up, and the host challenges them to play RPS simultaneously. 
If the host shows rock (paper, scissors) on stage, 
the contestants who show paper (scissors, rock, respectively) move on to the next round, while the other contestants are eliminated and asked to sit down. 
The competition continues until only one person remains standing, who is declared the winner and receives the prize. 
This method is widely used in Japan.

Board game designer Yoshiteru Shinohara\footnote{\url{https://boardgamegeek.com/boardgamedesigner/80468/yoshiteru-shinohara}} introduced a more engaging variation known as Greedy\footnote{This is because  greedy players are more willing to show paper at the risk of being eliminated immediately. } RPS, or 
Shinohara RPS in honor of its inventor.\footnote{In Japanese, Shinohara janken or Yokubari janken.} 
Before the game, the host explains the rule:
\begin{quote}
I will always show rock in every round.
This gives you a simple way to win: Just show paper. However, there is a catch to keep things interesting. If two or more players show paper in the same round, they are immediately eliminated.
So, you must decide: Do you play it safe by choosing rock, or take a risk with paper, hoping that no one else does the same?
\end{quote}
The game is played by at least three players and proceeds as follows.
\begin{enumerate}
	\item When three or more players remain, each round begins with the host calling out ``Rock-Paper-Scissors'' while demonstratively showing rock as the cue for Step 2. If exactly two players remain, they play classic RPS to determine the winner. 
	\item On the host’s cue, all remaining players simultaneously show rock or paper.
	\begin{itemize}
		\item If exactly one player shows paper, that player is declared the winner. 
		\item If two or more players show paper and two or more players show rock, the former are eliminated, and the latter move on to the next round.
		\item If exactly one player shows rock, that player is declared the winner.
		\item If all players show the same option (paper or rock), they repeat the round.\footnote{In the next section, where we formulate the game as an extensive game, we assume that the payoff for an infinite history with $n$ players remaining is $1/n$.} 
	\end{itemize}
\end{enumerate}

Shinohara RPS requires players to anticipate their opponents' choices. A player might choose rock by predicting that at least two others will choose paper, or choose paper in the expectation that all opponents will favor rock. 
The excitement of the game lies in outsmarting opponents through evolving strategies, which distinguishes Shinohara RPS from the purely random nature of classic RPS.

Shinohara RPS may also be of interest to game theory researchers since it is an extensive game that has not been studied in depth. This game raises nontrivial questions about how the equilibrium probabilities of paper and rock evolve as rounds progress. 
 Moreover, in the author's experience, incorporating Shinohara RPS into a game theory course can be beneficial.\footnote{See Appendix \ref{Classroom use}.} Having students play the game improves their understanding of key concepts such as strategic environments and mixed strategies. 
 It also promotes a lively and collaborative classroom atmosphere.

In light of the above, this paper introduces Shinohara RPS to researchers and studies its subgame perfect equilibria (SPE), where the expected payoff is the probability of winning.  
This game has a unique symmetric SPE, in which
the probability of choosing paper is the solution to the equation \((1-p)^{n-1} + p^{n-1}/n = 1/n\) (see Figure~\ref{fig1}) when $n$ players remain. 
The game also admits pure-strategy asymmetric SPE, including the following two types. In one type, one player chooses paper and the other players choose rock in each round. In another type, two players choose paper and the other players choose rock in each round. 
Furthermore, there is a continuum of SPE in which exactly one player adopts a mixed strategy.
The considerable variety of SPE in Shinohara RPS suggests that no particular equilibrium is likely to emerge consistently in actual play.

\begin{figure}[t]
\caption{The graph of the equilibrium probability of paper}
\begin{center} \includegraphics[height=7cm]{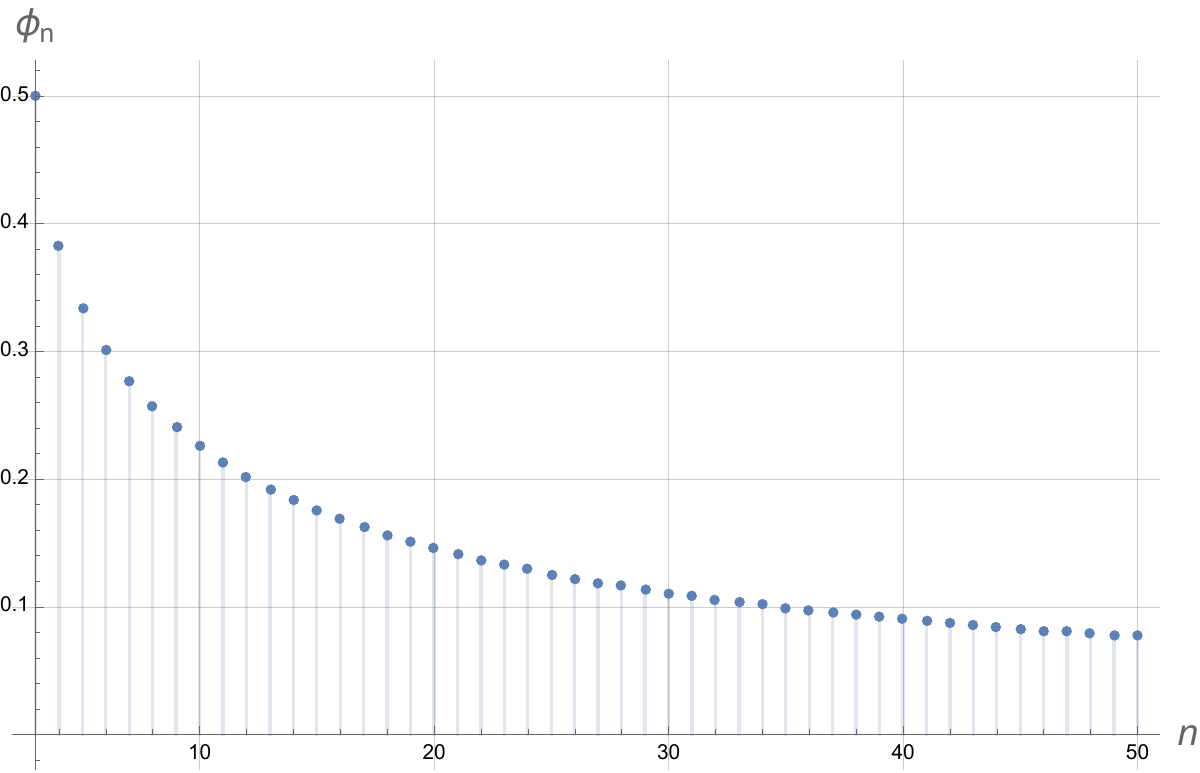}
\end{center}\label{fig1}
\end{figure}

\section{Equilibrium of Shinohara RPS}

Shinohara RPS is modeled as the following extensive game
with observable actions and simultaneous moves. 
Let $I=\{1,\ldots,|I|\}$ be the set of players with $|I|\geq 3$. 
An action of each player is either rock or paper. 
A history is a sequence of action profiles of surviving players. 
A terminal history is a finite history that ends with a single winner or with exactly two players remaining, or an infinite history in which the set of surviving players remains unchanged.
For each round, we call a player a loser if this player is eliminated or another player wins.
A strategy of a player is a function that assigns an action to each non-terminal history in which that player has not yet lost. 
All losers receive a payoff of $0$.  
If a single player wins, that player receives a payoff of $1$.  
If exactly two players survive, each receives a payoff of $1/2$.  
If the history is infinite and $n$ players survive, each receives a payoff of ${1}/{n}$.
Rational players aim to maximize the expected payoff  given the opponents' strategies.

\subsection{Symmetric subgame perfect equilibrium}

We derive a symmetric subgame perfect equilibrium (SPE) in which every player follows the same strategy. 
For each history $h$, let $p_h$ denote the equilibrium probability of paper in the first round after the history $h$. 
We must have $0<p_h <1$ for the following reason: If $p_h=0$ (or $p_h=1$), every player chooses rock (or paper), but one player can win by deviating to paper (or rock). 
For each nonterminal history $h$, let $n_h\ge 2$ denote the number of remaining players. 

Consider the subgame following a history $h$ with $n_h=n\geq 3$. 
Each player has the same chance of winning $1/n$.  
Thus, the probability of winning for a player who will choose paper in the first round of this subgame is also $1/n$ since the player is indifferent between paper and rock under a mixed-strategy SPE. 

This probability can also be calculated in a different way. Note that the player who chooses paper wins if and only if one of the following events occurs.
\begin{enumerate}[A.]
\item All the other players choose rock in the first round, which results in this player winning.
\item (i) All the other players choose paper in the first round, and (ii) in the subsequent subgame, this player wins.
\end{enumerate}
Thus, the probability of winning for this player is the probability of either A or B occurring. 
The probability of A is equal to $(1-p_h)^{n-1}$. 
The probability of B is equal to the probability of B-(i) multiplied by the probability of B-(ii). The former is equal to $p_h^{n-1}$, and the latter is equal to $1/n$. 
Therefore, the probability of winning for the player choosing paper is $(1-p_h)^{n-1}+p_h^{n-1}/n$. 

Combining the above arguments, it follows that $p_h$ is a solution to the equation 
\begin{equation}
(1-p_h)^{n-1}+\frac{p_h^{n-1}}{n}=\frac{1}{n},
\label{prob1}
\end{equation}
which can be rewritten as 
\begin{equation}\frac{(1-p_h)^{n-2}}{1+\sum_{k=1}^{n-2}p_h^{k}}=\frac{1}{n}.\label{prob2}
\end{equation}
We can make two observations from this equation. 
First, a solution is determined solely by the number of remaining players $n$. 
Second, a solution in the open unit interval is unique since the left-hand side of \eqref{prob2} is decreasing in $p_h$ with the minimum zero and the maximum one. 
Let $\phi_n\in (0,1)$ denote the unique solution,  
which is plotted in Figure \ref{fig1} and listed in Table \ref{table1}.

Therefore, if $p_h$ is an equilibrium probability of paper, then it must be equal to $\phi_n$. 
The converse is also true: If players choose paper with probability $p_h=\phi_n$ for every history $h$, then they are indifferent between paper and rock, so it must be an equilibrium probability of paper. 
That is, we obtain the following proposition.

\begin{proposition}\label{main result}
There exists a unique symmetric SPE, in which the probability of choosing paper is given by $\phi_n$ for each $n\geq 3$. 
\end{proposition}

\begin{table}[t]
 \caption{The equilibrium probability of paper}
\begin{tabular}{ccccccccccc}
\hline
$n$ & & & 3 & 4 & 5 & 6 & 7 & 8 & 9 & 10  \\
$\phi_n$ & & & 0.500 & 0.382 & 0.333 & 0.302 & 0.277 & 0.257 & 0.240 & 0.226\\\hline
$n$ & 11 & 12 & 13 & 14 & 15 & 16 & 17 & 18 & 19 & 20  \\
$\phi_n$  & 0.213 & 0.202 & 0.192 & 0.184 & 0.176 & 0.169 & 0.162 & 0.156 & 0.151 & 0.146 \\\hline
$n$ & 21 & 22 & 23 & 24 & 25 & 26 & 27 & 28 & 29 & 30  \\
$\phi_n$ & 0.141 & 0.137& 0.133 & 0.129 & 0.126 & 0.122 & 0.119 & 0.116 & 0.113 & 0.111 \\\hline
$n$ & 31 & 32 & 33 & 34 & 35 & 36 & 37 & 38 & 39 & 40  \\
$\phi_n$ & 0.108 & 0.106 & 0.104 & 0.101 & 0.099 & 0.097 & 0.095 & 0.094 & 0.092 & \
0.090\\\hline
$n$ & 41 & 42 & 43 & 44 & 45 & 46 & 47 & 48 & 49 & 50  \\
$\phi_n$ & 0.089 & 0.087 & 0.086 & 0.084 & 0.083 & 0.083 & 0.080 & 0.079 & 
0.078 & 0.077\\\hline
\end{tabular}\label{table1}
\end{table}

\subsection{Asymmetric Markov perfect equilibrium}

In Shinohara RPS, each history determines the set of remaining players $N\subseteq I$, which is a state of the game. 
A strategy is referred to as a Markov strategy if it depends on a history only through the state, and an SPE composed of Markov strategies is referred to as a Markov perfect equilibrium (MPE) \citep{fudenberg1991}. 
The symmetric SPE described above is one example of an MPE.

We demonstrate the existence of a continuum of asymmetric MPE.
To facilitate the exposition of Markov strategies in a round,
we refer to a player who chooses paper with probability $p\in [0,1]$ as a type-$p$ player. 

First, 
consider a Markov strategy profile in which, for each state $N$ with $|N| \geq 3$:  
\begin{itemize}
    \item Exactly one player is of type-$1$. 
    \item All other players are of type-$0$. 
\end{itemize}
Such a Markov strategy profile is an MPE because   
no player has an incentive to deviate, regardless of the selection of a type-$1$ player.
\begin{itemize}
    \item The type-$1$ player wins with probability one. Since switching to rock does not increase the winning probability, she has no incentive to deviate.  
    \item Each type-$0$ player loses with probability one. If a type-$0$ player switches to paper, this player still loses with probability one. Thus, there is no incentive to deviate.  
\end{itemize}

Next, consider another Markov strategy profile in which, for any $N$ with $|N| \geq 3$:  
\begin{itemize}
    \item Exactly two players are of type-$1$. 
    \item All other players are of type-$0$. 
\end{itemize}
Such a Markov strategy profile is also an MPE, regardless of the selection of type-$1$ players.

\begin{itemize}
    \item Each type-$1$ player loses with probability one. Even if a type-$1$ player switches to rock, she still loses with probability one, in which case the other type-$1$ player wins. Thus, there is no incentive to deviate.  
    \item If $|N|\geq 4$, each type-$0$ player moves on to the next round, but switching to paper results in immediate elimination. If $|N|=3$, there is exactly one type-$0$ player, who wins with probability one. 
 In both cases, a type-$0$ player has no incentive to deviate.  
\end{itemize}

Finally, consider a convex combination of the above Markov strategy profiles, 
 in which, for any $N$ with $|N| \geq 3$:  
\begin{itemize}
    \item Exactly one player is of type-$1$. 
    \item Exactly one player is of type-$q_N$ with $q_N\in [0,1]$. 
    \item All other players are of type-$0$. 
\end{itemize}
Such a Markov strategy profile is also an MPE, regardless of the selection of type-$1$ and type-$q_N$ players, for any $q_N\in [0,1]$.
\begin{itemize}
    \item The best response of the type-$1$ player is to choose paper in both cases where the type-$q_N$ player chooses rock or paper, just as in the case of a type-$1$ player in the above MPE. Thus, the type-$1$ player has no incentive to deviate for arbitrary $q_N\in [0,1]$.  
    \item The best response of a type-$0$ player is to choose rock in both cases where the type-$q_N$ player chooses rock or paper, just as in the case of a type-$0$ player in the above MPE. Thus, a type-$0$ player has no incentive to deviate for arbitrary $q_N\in [0,1]$.    
    \item The type-$q_N$ player is indifferent between rock and paper since both choices result in a loss with probability one. Thus, there is no incentive to deviate.  
\end{itemize}
This implies the following proposition. 

\begin{proposition}\label{main result 3}
There exists a continuum of asymmetric MPE. 
\end{proposition}

\newpage

\begin{center}
{\Large\bf Appendix}
\end{center}

\appendix
\renewcommand{\theequation}{\alph{equation}}
\setcounter{equation}{0}

\renewcommand{\theproposition}{\Alph{proposition}}
\setcounter{proposition}{0}  

\section{Mixed-strategy Markov perfect equilibrium}

The unique symmetric SPE is a totally mixed-strategy MPE.
We conjecture that no other totally mixed-strategy MPE exists, which remains unproven.
To facilitate future research, this appendix derives a system of equations characterizing a totally mixed-strategy MPE.

A totally mixed Markov strategy profile is given by the collection of probabilities 
\[
\{\pi_{iN}\in (0,1)\mid \text{$i\in N\subseteq I$, $|N|\geq 3$}\},
\]
where $\pi_{iN}$ is the probability of choosing paper for player $i\in N$ in a round with a state $N$. 
Given a totally mixed Markov strategy profile $\{\pi_{iN}\}$, consider a subgame with a state $N$. 
For each $i\in N$, we denote by $\rho_{iN}\in(0,1)$ the probability that player $i$ wins in the subgame. 
Note that $\sum_{i\in N}\rho_{iN}=1$. 

When player $i$ chooses paper in the first round, her winning probability is 
\begin{equation}
\prod_{j\in N_{-i}}(1-\pi_{jN})+\rho_{iN}\prod_{j\in N_{-i}}\pi_{jN},\label{winning prob paper}
\end{equation}  
where $N_{-i}=N\setminus\{i\}$. This follows from an argument similar to that used for the symmetric SPE. 
Player $i$ who chooses paper wins if and only if all the other players choose rock in the first round, with the probability given by the first term, or all the other players choose paper in this round, followed by player $i$'s victory in the subsequent subgame with the same state $N$, with the probability given by the second term.

When player $i$ chooses rock in the first round, 
her winning probability is 
\begin{equation}
\prod_{j\in N_{-i}}\pi_{jN}+\rho_{iN}\prod_{j\in N_{-i}}(1-\pi_{jN})
+\sum_{S\subset N:\, i\in S,\ 2\leq |S|\leq |N|-2}\rho_{iS}\prod_{j\in S_{-i}}(1-\pi_{jN})\prod_{j\in N\setminus S}\pi_{jN},\label{winning prob rock}
\end{equation}  
where $S_{-i}=S\setminus\{i\}$. 
The first term is the probability that all the other players choose paper in the first round. 
The second term is the probability that all the other players choose rock in this round and then player $i$ wins in the subsequent subgame with the same state $N$.
Each term in the final summation is  
the probability that players in $S$, which includes player $i$, choose rock and 
players in $N\setminus S$ choose paper in this round and then player $i$ wins in the subsequent subgame with the state $S$.

Therefore, player $i$ is indifferent between paper and rock if and only if \eqref{winning prob paper} and \eqref{winning prob rock} are equal to $\rho_{iN}$. 
This leads to the following characterization of a totally mixed-strategy MPE.

\begin{proposition}\label{main result 2}
A totally mixed Markov strategy profile $\{\pi_{iN}\}$ is an MPE if and only if it satisfies the following system of equations: For all $N\subseteq I$ with $|N|\geq 3$ and all $i\in N$,
\begin{equation}
\frac{\prod_{j\in N_{-i}}(1-\pi_{jN})}{1-\prod_{j\in N_{-i}}\pi_{jN}}=
\frac{\prod_{j\in N_{-i}}\pi_{jN}+\Lambda_{N}}{1-\prod_{j\in N_{-i}}(1-\pi_{jN})},
\label{rho equal}
\end{equation}
where 
\begin{equation}
\Lambda_{N}=
\sum_{S\subset N:\, i\in S,\ 2\leq |S|\leq |N|-2}\rho_{iS}\prod_{j\in S_{-i}}(1-\pi_{jN})\prod_{j\in N\setminus S}\pi_{jN},\label{lambda eq}	
\end{equation}
\begin{equation}
\rho_{iS}=
\begin{cases}
\displaystyle\frac{\prod_{j\in N_{-i}}(1-\pi_{jS})}{1-\prod_{j\in S_{-i}}\pi_{jS}} &\text{ if }|S|\geq 3,\\
1/2 & \text{ if }|S|=2.\\
\end{cases}
\label{rho eq}		
\end{equation}
 \end{proposition}

\begin{proof}
Suppose that $\{\pi_{iN}\}$ is an MPE. 
Recall that classic RPS is played when the number of remaining players is two, so $\rho_{iS}=1/2$ if $|S|=2$. 

For any $N$ with $|N|\geq 3$, \eqref{winning prob paper} is equal to $\rho_{iN}$. 
Solving this equation for $\rho_{iN}$ gives 
\begin{equation}
\rho_{iN}=
\frac{\prod_{j\in N_{-i}}(1-\pi_{jN})}{1-\prod_{j\in N_{-i}}\pi_{jN}},\label{rho1}
\end{equation}  
which implies \eqref{rho eq}.
In addition, \eqref{winning prob rock} is equal to $\rho_{iN}$. 
Solving this equation for $\rho_{iN}$ gives 
\begin{align}
\rho_{iN}=
\frac{\prod_{j\in N_{-i}}\pi_{jN}+\Lambda_{N}}{1-\prod_{j\in N_{-i}}(1-\pi_{jN})}.\label{rho3}
\end{align}  
Then, \eqref{rho equal} is implied by 
\eqref{rho1} and \eqref{rho3}. 

Conversely, if $\{\pi_{iN}\}$ satisfies the system of equations, every player in every state is indifferent between paper and rock by the above discussion, so it is an MPE. 
\end{proof}

Since the symmetric SPE is a totally mixed-strategy MPE, the system of equations has a solution: $\pi_{iN}=\phi_{|N|}$ for all $N\subseteq I$ and $i\in N$. 
Whether another solution exists is an open question.

If $|N|= 3$, we have  
$\Lambda_N=0$. In this case, 
we can directly show that the system of equations has a unique solution $\pi_{iN}=\phi_{|N|}$ for all $i\in N$.
Thus, one might attempt to prove that $\pi_{iN}=\phi_{|N|}$ holds for all $N$ and $i\in N$ by induction, assuming the induction hypothesis that $\rho_{iS}=1/|S|$ for all $S$ with $|S|\leq n-1$.  
This hypothesis simplifies \eqref{rho equal} to an equation with $(\pi_{iN})_{i\in N}$ as the only variables. 
However, this equation remains quite complicated due to $\Lambda_N\neq 0$, making the proof challenging. 

\section{Classroom use of Shinohara RPS}\label{Classroom use}

This appendix illustrates the pedagogical use of Shinohara RPS by explaining how the author incorporates it into a game theory course.

I introduce Shinohara RPS on the first day of an undergraduate game theory course or during the game theory unit in a microeconomics course. It serves as an entry point for presenting game theory as a method for analyzing strategic environments, where decision-makers' optimal choices depend on their beliefs about others' choices. I emphasize that such situations are ubiquitous in social and economic life, and that the ability to anticipate and understand others' beliefs is essential in these contexts. I also note that this ability is studied in other disciplines, such as psychology and philosophy, under the concept of theory of mind.

To illustrate strategic environments in an engaging way, I explain the rules of Shinohara RPS and immediately have the students play the game, with myself acting as the host. At the start of the game, all students are asked to stand. After each round, those who are eliminated must sit down. The number of standing participants decreases with each round, creating a visible sense of progress. The game accommodates the entire class simultaneously, making it suitable even for large, lecture-based courses without teaching assistants. It also works well in smaller seminar settings, provided there are at least four students.

Once a winner is determined, I ask the class to reflect: What were you thinking when you made your choice? How did you predict what others would do? These discussions are always lively, as students are eager to explain their reasoning. In smaller classes, their explanations often refer to the perceived personalities or behavioral tendencies of their classmates. Through this process, students naturally experience the essence of strategic reasoning and belief formation.

Later in the course, I use both classic RPS and the three-player Shinohara RPS to illustrate mixed-strategy Nash equilibrium. When introducing extensive games, I return to Shinohara RPS and its SPE, demonstrating how a simple classroom game can be rigorously analyzed.

\bigskip\bigskip\medskip

\noindent{\Large\bf Conflict of Interest}\bigskip

\noindent 
The author declares no conflicts of interest regarding this manuscript.


\begin{thebibliography}{Fudenberg and Tirole(1991)}
	

\bibitem[Fudenberg and Tirole(1991)]{fudenberg1991} Fudenberg, D., Tirole, J., 1991. Game Theory. MIT Press.

\end{thebibliography}
\end{document}